\documentclass[11 pt]{article}
\usepackage{graphicx} 

\usepackage[square,numbers]{natbib}

\title{Verifying global identifiability   of parametric linear ODE models is NP-hard\footnote{This research was partially supported by the NSF grants CCF-2212460 and DMS-1853650.}}
\author{Alexey Ovchinnikov and Pedro Soto}
\date{}

\usepackage{times,pslatex}
\usepackage{color,bm,xcolor}
\usepackage{amsthm}
\usepackage{amsfonts, comment}
\usepackage{amssymb}
\usepackage{amsmath}
\usepackage[colorlinks, linkcolor=reference,
 citecolor=citation,
         urlcolor=e-mail]
      {hyperref}
\definecolor{conditional}{rgb}{0,1,0}
\definecolor{e-mail}{rgb}{0,.40,.80}
\definecolor{reference}{rgb}{.20,.60,.22}
\definecolor{mrnumber}{rgb}{.80,.40,0}
\definecolor{citation}{rgb}{.20,.60,.22}

\DeclareMathOperator{\Char}{char}
\DeclareMathOperator{\Wr}{Wr}

\newtheorem{theorem}{Theorem}
\newtheorem{corollary}{Corollary}

\newtheorem{lemma}{Lemma}
\theoremstyle{definition}
\newtheorem{definition}{Definition}
\theoremstyle{remark}
\newtheorem{example}{Example}
\newtheorem{remark}{Remark}

\begin{document}

\maketitle

\begin{abstract}
Global parameter identifiability is a property of a parametric ODE model to recover the parameter values uniquely from the input-output data. Not all parametric ODE models have this property, and checking for parameter identifiability is a prerequisite to perform numerical parameter estimation. There are many algorithms and software packages for global parameter identifiability, and frequently the runtime is large. However, the computational complexity for this problem has not been analyzed yet, though there are complexity results for local (finitely many values fit the data) parameter identifiability. In this paper, we estimate the complexity of checking global parameter identifiability over real fields for ODE models that depend linearly on the state variables and rationally on the parameters. In particular, we prove that it is equivalent to the injectivity problem.
\end{abstract}

\section{Introduction}
Parametric ODEs are widely used in modeling, and estimating the numerical parameter values from measured data is often needed because parameters frequently have a useful physical meaning. However, if multiple (finitely or infinitely many) values of the parameters fit the same data, it is no longer possible to solve the estimation problem without additional insights. Global parameter identifiability is a property of an ODE model to recover the parameters uniquely from the data, and there are many ODE models used in practice with non-identifiabile (infinitely many parameter values fit the data) or only locally identifiable (finitely but more than one parameter values fit the data) parameters.

The global parameter identifiability problem has been studied for several decades, and multiple algorithms and software packages have been developed, including, e.g., DAISY~\cite{DAISY}, SIAN~\cite{SIAN}, 
COMBOS~\cite{COMBOS},  Structural Identifiability Toolbox~\cite{Ilmer2021}, and StructuralIdentifiability.jl~\cite{structidjl}. There has been effort to improve the running time and memory use~\cite{HOPY2020,structidjl,trbasis,weights} for global identifiablity checking. However, computational complexity results for global identifiability are still missing. The existing results are limited to the complexity of local identifiability~\cite{Sedoglavic} and ODE models over complex numbers.

In this paper, we study the computational complexity of the parameter identifiability problem for linear ODE systems that depend rationally on its parameters, over a more general class of fields than just the field of complex numbers: fields of characteristic zero and finite fields. This includes the field of real numbers, which is crucial in applications, but was largely avoided in the existing theoretical analysis for global identifiability. In particular, we prove that it is NP-hard by showing that the injectivity problem (known to be NP-hard) is polynomial-time Karp reducible to global identifiability. Conversely, we show that global identifiability is polynomial-time Karp reducible to an unexpanded form of the injectivity problem, which we prove lies in co-$\mathrm{NP}_R$ in the Blum–Shub–Smale model. Together these give that, in the algebraic complexity model, both global identifiability and IO-identifiability are NP-hard and lie in co-$\mathrm{NP}_R$ (see Theorem~\ref{thm:main} and Corollary~\ref{cor:final} for more details).

The paper is organized as follows. We introduce the necessary terminology from differential algebra and indentifiability  in Section~\ref{sec:defs}. In Section~\ref{sec:main}, we state our main complexity result, which we then prove in Sections~\ref{sec:lower} and~\ref{sec:upper}.

\section{Definitions}\label{sec:defs}
We consider linear ODE models of the form
\begin{equation}\label{eq:sys}
\Sigma = 
\begin{cases}
\bm x' = \bm f(\bm x,\bm\alpha)\\
\bm y = \bm g(\bm x, \bm\alpha),
\end{cases}
\end{equation}
where $\bm f$ and $\bm g$ are vectors of 
functions over a field $\mathbb{F}$ (of any characteristic) that are rational in the parameters $\bm \alpha = (\alpha_1,\ldots,\alpha_s)$ and in the state variables $\bm x = (x_1,\ldots,x_n)$, with output functions $\bm y =(y_1,\ldots,y_m)$. 
We now introduce basic definitions from commutative and differential algebra that we use in the paper:
\begin{enumerate}
  \item A {\em differential ring} $(R,')$ is a commutative ring with a derivation $':R\to R$, that is, a map such that, for all $a,b\in R$, $(a+b)' = a' + b'$ and $(ab)' = a' b + a b'$. 
  \item The {\em ring of differential polynomials} in the variables $x_1,\ldots,x_n$ over a differential field $\mathbb{F}$ is the ring $\mathbb{F}\left[x_j^{(i)}\mid i\geqslant 0,\, 1\leqslant j\leqslant n\right]$ with a derivation defined on the ring by ${\left(x_j^{(i)}\right)}' := x_j^{(i + 1)}$. 
  This differential ring is denoted by $\mathbb{F}\{x_1,\ldots,x_n\}$.
  \item An ideal $I$ of a differential ring $(R,\delta)$ is called a {\em differential ideal} if, for all $a \in I$, $\delta(a)\in I$. For $F\subset R$, the smallest differential ideal containing the set $F$ is denoted by $[F]$.

    \item An ideal $P$ in a commutative ring $R$ is said to be {\em prime} if, for all $a, b \in R$, $a\cdot b \in P$ implies $a\in P$ or $b \in P$.
\item We say that a differential field $(\mathbb{F}, ')$ is a field of constants if, for all $a \in \mathbb{F}$, $a' = 0$.
    \item Let $(\mathbb{F}, ')$ be a field of constants. Given $\Sigma$ as in~\eqref{eq:sys}, we define the differential ideal of $\Sigma$ as $I_\Sigma=[\bm{x}'-\bm{f},\bm{y}-\bm{g}] \subset \mathbb{F}(\bm{\alpha})\{\bm{x},\bm{y}\}$. Following the proof of  \cite[Lemma~3.2]{HOPY2020}, $I_\Sigma$ is a prime differential ideal as the results from polynomial ideals used in the proof of \cite[Lemma~3.2]{HOPY2020} did not use characteristic zero.
    \item  The image of $(\bm x,\bm y)$ under the canonical homomorphism $\mathbb{F}(\bm{\alpha})\{\bm{x},\bm{y}\} \to \mathbb{F}(\bm{\alpha})\{\bm{x},\bm{y}\}/I_\Sigma$ is called the {\em generic solution} of $\Sigma$.
    \item  For differential fields $\mathbb{F} \subset \mathbb{K}$ and elements $a_1,\ldots,a_n$, the smallest differential subfield of $\mathbb{K}$ that contains $\mathbb{F}$ and $a_1,\ldots,a_n$ is denoted by $\mathbb{F}\langle a_1,\ldots,a_n\rangle$.

    \item A field $\mathbb{F}$ is called {\em formally real}, if there is no way to express $-1$ as a sum of squares, \emph{i.e.,} $a_1^2 + ... + a_n ^2 \neq 0$, for all $a_1,...,a_n \in \mathbb{F} \setminus \{0\}$. Any such field must contain the rationals $\mathbb{Q}$ since $1$ is a square and thus $1+...+1 = n \neq 0$ for all $n \in \mathbb{N}_{\geq 0}$ by induction.
    \item For a field $\mathbb{F}$ of characteristic $p$, let $\mathbb{F}^{p^{-\infty}}$ denote the {\em perfect closure} of $\mathbb{F}$. Note that, if $\Char \mathbb{F} =0$ or $\mathbb{F}$ is finite, then \[\mathbb{F}^{p^{-\infty}} = \mathbb{F}.\] If $\Char \mathbb{F} = p \ne 0$ and $\alpha$ is transcendental over $\mathbb{F}$, then \[\mathbb{F}(\alpha)^{p^{-\infty}} =\mathbb{F}^{p^{-\infty}}\left(\alpha,\alpha^{1/p}, \alpha^{1/p^2},\ldots\right) .\]
\end{enumerate}

\begin{definition}[{Input-output (IO) identifiability, cf.\cite[Definition~2]{ident-compare}}]
\label{item:eight} Let $p =\Char\mathbb{F}$ and $k$  the smallest subfield such that $\mathbb{F} \subset k\subset \mathbb{F}(\bm\alpha)$ and $I_\Sigma \cap \mathbb{F}(\bm\alpha)\{\bm y\}$ is generated as an ideal by $I_\Sigma \cap k\{\bm y\}$. A parameter $\alpha$ is said to be {\em IO-identifiable} over a subset $D \subset \mathbb{F}$ if there exists a rational function $h : \mathbb{F}\to \mathbb{F}$ such that $h(\alpha) \in k^{p^{-\infty}}$ and $h$ is injective on $D$. 

The tuple of parameters $\bm\alpha$ is said to be {\em IO-identifiable} on a tuple of subsets $\bm D$ if, for all $i$, $\alpha_i$ is identifiable on $D_i$.
\end{definition}

\begin{remark} The name IO-identifiability in Definition~\ref{item:eight} comes from the fact that the field $k$ in the definition is generated by the coefficients of so-called IO-equations, which can be computed using differential algebra algorithms, see~\cite[Section~5]{ident-compare}. This provides a convenient way of checking IO-identifiability.
\end{remark}

\begin{definition}[Identifiability, {cf. \cite[Sections~2 and 3]{GR2021}}]\label{def:id}
 A parameter $\alpha$ is said to be {\em identifiable}  in system~\eqref{eq:sys} over a subset $D \subset\mathbb{F}$ if, for every generic solution $(\bm x^\ast, \bm y^\ast)$ of~\eqref{eq:sys}, there exists $h \in\mathbb{F}\langle \bm y^\ast\rangle^{p^{-\infty}}\cap \mathbb{F}(\alpha)$ such that the rational map $h : \mathbb{F} \to \mathbb{F}, \alpha \mapsto h(\alpha),$ is injective on $D$.

 The tuple of parameters $\bm\alpha$ is said to be {\em identifiable} on a tuple of subsets $\bm D$ if, for all $i$, $\alpha_i$ is identifiable on $D_i$.
\end{definition}

\begin{remark} If $\mathbb{F}=\mathbb{C}$ in Definition~\ref{def:id} and $D = \mathbb{C}$, then the definition becomes equivalent to the standard definition of identifiability (see e.g.~\cite[Definition~7]{allident}) because one can show that $h \in \mathbb{C}(\alpha)$ is injective as a map $\mathbb{C} \to \mathbb{C}$ if and only if $\alpha \in \mathbb{C}(h)$. 
\end{remark}

We will now give two clarifying examples if $\mathbb{F} = \mathbb{R}$, which is the most frequent case occuring in the applications.

\begin{example}
We will see that, over $\mathbb{R}$, an identifiable parameter $\alpha$ does not have to belong to the field $\mathbb{R}\langle y^\ast\rangle$, unlike the case of $\mathbb{C}$ as in \cite[Definition~7]{allident}. For this, consider the ODE model
\begin{equation}\label{eq:alphacubed}
\begin{cases}
x' = \alpha^3x,\\
y = x.
\end{cases}
\end{equation}
Then, for any generic solution $(x^\ast,y^\ast)$ of~\eqref{eq:alphacubed}, \[h := \alpha^3 = \frac{{y^\ast}'}{y^\ast} \in \mathbb{R}\langle y^\ast\rangle.\] Since the map $\alpha \mapsto \alpha^3$ is injective on $\mathbb{R}$, the parameter $\alpha$ is identifiable according to Definition~\ref{def:id}. However, $\alpha \notin \mathbb{R}\langle y^\ast\rangle$. Moreover, $\alpha \notin \mathbb{C}\langle y^\ast\rangle$, and so $\alpha$ is not identifiable over $\mathbb{C}$.
\end{example}

\begin{example} The notion of identifiability over $\mathbb{R}$ has a user-specified choice of a domain $D$ to take into account more possible situations. To illustrate this, consider the ODE model
\[
\begin{cases}
x' = (\alpha^3 -\alpha)x,\\
y=x
\end{cases}
\]
and the following intervals \[D_1 = (0,1),\ \ D_2 = (1,2),\ \ D_3 = (-\infty,+\infty).\] We have that \[\mathbb{R}\langle y^\ast\rangle\cap \mathbb{R}(\alpha) = \mathbb{R}(\alpha^3-\alpha).\] Since the function $\alpha^3-\alpha$ is not injective on $(0,1)$ or $(-\infty,+\infty)$, no $h\in\mathbb \mathbb{R}(\alpha^3-\alpha)$ is injective there, and so $\alpha$ is not identifiable on $D_1$ or $D_3$. On the other hand, $\alpha^3-\alpha$ is injective on $(1,2)$, and so $\alpha$ is identifiable on $D_2$.
\end{example}

The following toy example shows why we take perfect closures in Definitions~\ref{item:eight} and~\ref{def:id} if $\Char{\mathbb{F}} = p > 0$.
\begin{example} Consider the model
\[
\begin{cases}
x' = \alpha^2x\\
y=x
\end{cases}
\]
with $\mathbb{F} = \mathbb{Z}_2$. Following Definition~\ref{item:eight}, we find the ideal \[I_\Sigma \cap\mathbb{Z}_2(\alpha)\{y\} = [y' - \alpha^2y],\] and so $k = \mathbb{Z}_2(\alpha^2)$. Since the map $\alpha \to \alpha^2$ is injective in $\mathbb{Z}_2$, it makes sense to call $\alpha$ IO-identifiable.  Similarly, considering Definition~\ref{def:id}, we see that 
\[
\alpha^2 = \cfrac{{y^\ast}'}{y^\ast} \in \mathbb{Z}_2\langle y^\ast\rangle\quad\text{and}\quad \alpha \notin \mathbb{Z}_2\langle y^\ast\rangle.
\] 
However, for all values $\alpha^\ast \in \mathbb{Z}_2$, we have $\alpha^\ast = {(\alpha^\ast)}^2$. Hence, in $\mathbb{Z}_2$, $\alpha^\ast$ is (trivially) uniquely determined by ${(\alpha^\ast)}^2$. This is consistent with taking the perfect closure as in Definition~\ref{def:id}:
\[
\mathbb{Z}_2\langle y^\ast\rangle^{p^{-\infty}}\cap \mathbb{Z}_2(\alpha)=\mathbb{Z}_2(\alpha),
\]
and so we say that $\alpha$ is globally identifiable.
\end{example}

\begin{definition}(The Class $\mathbf{NP}$, {cf. \cite[Definition 2.1]{arora2009computational}}) 
A language $L \subseteq\{0,1\}^*$ is in $\mathbf{NP}$ if there exists a polynomial $p: \mathbb{N} \rightarrow \mathbb{N}$ and a polynomial-time Turing machine $M$ such that, for every $x \in\{0,1\}^*$,
\[
x \in L \iff \exists u \in\{0,1\}^{p(|x|)} \text { s.t. } M(x, u)=1
\]

If $x \in L$ and $u \in\{0,1\}^{p(|x|)}$ satisfy $M(x, u)=1$ then we call $u$ a {\em certificate} (or {\em witness}) for $x$ (with respect to the language $L$ and machine $M$ ).
\end{definition}

\begin{definition}(Polynomial-time reducible, {cf. \cite[Definition 2.7]{arora2009computational}})
\begin{itemize}
\item 
We say that a language $A \subseteq\{0,1\}^*$ is \textit{polynomial-time Karp reducible} (abbreviated as \textit{polynomial-time reducible}) to a language $B \subseteq\{0,1\}^*$, denoted by $A \leq_p B$, if there is a polynomial-time computable function $f:\{0,1\}^* \rightarrow\{0,1\}^*$ such that for every $x \in\{0,1\}^*$,
\[ x \in A \iff f(x) \in B.\]
\item 
We say that $B$ is $\mathbf{N P}$-\textit{hard} if $A \leq_p B$ for every $A \in \mathbf{N P}$. 
\item We say that $B$ is $\mathbf{N P}$-\textit{complete} if $B$ is $\mathbf{N P}$-hard and $B \in \mathbf{N P}$.
\end{itemize}
\end{definition}

\begin{definition}\label{def:5}(Algebraic $\mathbf{NP}_R$ or the BSS-model, {cf. \cite[Chapter 5, Definition 1]{blum1998complexity}})
Given a ring, $R$, we give an informal characterization of $\mathbf{P}_R$, $\mathbf{NP}_R$, $\mathbf{N P}_R$-hard, etc. 
The definitions remain almost the same as before with the only difference being: 
\begin{itemize}
    \item We consider more general languages of the form $L \subset R^*$. 
    \item Our machines are further allowed to perform basic arithmetic operations, \emph{e.g., }$+,*$, at unit cost. 
\end{itemize}
\end{definition}
\begin{remark}
    The standard Turing complexity is given by setting $R:= \mathbb{F}_2$ and is, therefore, the most natural algebraic generalization of Turing complexity.
    However, there are many other algebraic computation models, see Chapter 16 of \cite{arora2009computational}.
Our definition of the BSS model is closer to the one in \cite{arora2009computational} than the one in \cite{blum1998complexity} which models a computation as a path in a finite directed graph. 
In practice, counting the number of arithmetic operations suffices to bound the algebraic complexity (from above). 
\emph{However, one must be very careful not to conflate uniform and non-uniform models of computation; in particular, in the case of $R:= \mathbb{F}_2$, it is paradoxically possible that $\mathbf{NP} \neq \mathbf{P} $ and, simultaneously, that there exists efficient polynomial size Boolean circuits (\emph{i.e.,} straightline programs over $R:= \mathbb{F}_2$) that solve every problem in $\mathbf{NP}  $!} 
The confusion disappears once one contemplates the following two facts 1) there are boolean circuits that solve the halting problem (therefore, boolean circuits are far more powerful than Turing machines) 2) simply because there exists an efficient circuit does not mean we have an efficient means by which to construct the circuits themselves. 
\end{remark}

\begin{definition}(Definable, {cf. \cite[Definition 1.3.1]{Marker_book}})
    Given a ring $R$ (respectively an ordered ring $R$), we say that $D$ is $R$-definable (resp. $R$-order-definable) if there is a formula $\varphi(\bm x, \bm y)$ in the language of rings (resp. language of ordered rings) and a tuple $\bm b \in R^ k$ such that 
    \begin{equation*}
        D = \{ \bm x \in R^n \mid \varphi(\bm x, \bm b) \}.
    \end{equation*}
\end{definition}

\begin{definition}\label{def:poly_def}
    We say that $D$ is {\em polynomial-time definable} if there exists a polynomial-time machine $M$ such that 
    $$
x \in D \iff M(x,D) = 1.
    $$
    In particular, 
    we have $L_\text{poly\_def} \in \mathbf{P}_R$.
    
\end{definition}

\begin{remark}
    To make the proceeding language, Definition~\ref{def:poly_def}, well-defined, we replace $D$ with its code, $|D|$, given some G\"odel-coding, $|\_|$, of the well formed formula over the language of rings. See   \cite[1.10 and 7.8]{10.5555/28907}.
\end{remark}

\section{Main result}\label{sec:main}
Let $\mathbb{F}$ be a field (any characteristic).
We define the language of injective mappings as
\[
L_\mathrm{inj} : = \{ (\bm g, D)  \mid  g_i \in \mathbb{F}[\alpha_1,...,\alpha_s],  \bm{g}: D \rightarrow \mathbb{F}^n \text{ is injective, } D \in L_\text{poly\_def}\} ,
\]
and let the language of globally identifiable/IO-identifiable ODE systems be defined as 
\begin{gather*}
L_\mathrm{gid} : = \{ ( \Sigma, D ) \mid  \text{the parameters $\bm\alpha$ are identifiable on } \bm D\}\\
L_{\mathrm{gid},1} : = \{  (\Sigma ,D ) \in L_\mathrm{gid}  \mid  \text{ with $m=1$}  \}\\
L_\mathrm{gioid} : = \{  (\Sigma,D ) \mid  \text{the parameters $\bm\alpha$ are IO-identifiable on } \bm D\}
\end{gather*}
Finally let the knapsack language be defined as 
\begin{equation}\label{eq:3}
    L_\mathrm{knap} := 
    \left \{ S = \{s_1,...,s_{n+1}\} \subset \mathbb{Z} \mid  \sum_{i \in [n]}s_i = s_{n+1} \right \}.
    \end{equation}
    \begin{remark}
The language $L_\mathrm{inj\_det}$ appearing in Theorem~\ref{thm:main} refers to the case in which each component $g_i$ of the input map is presented as a single polynomial in expanded form (a list of coefficients indexed by monomial degrees). The unexpanded language $L_\mathrm{inj\_unexp}$ will be introduced in Section~\ref{sec:upper} and generalizes this by allowing each component to be presented recursively as a composition of polynomial maps. Lemma~\ref{lem:not_too_hard} will show that $L_\mathrm{inj\_unexp} \in \mathbf{co\text{-}NP}_R$; since $L_\mathrm{inj\_det}$ is the depth-one case of $L_\mathrm{inj\_unexp}$, it follows a fortiori that $L_\mathrm{inj\_det} \in \mathbf{co\text{-}NP}_R$. This is not vacuous because the unexpanded representation can be exponentially more compact than the expanded one (Example~\ref{ex:full} gives an $O(n)$-size unexpanded input whose expansion has $n!$ terms), yet the $\mathbf{co\text{-}NP}_R$ bound applies uniformly to both presentations. Finally, if $\mathbb{F}$ is algebraically closed, the proof of Lemma~\ref{thm:gid_hard} goes through unchanged (the perfect closure is trivial and the Galois-theoretic argument simplifies), so Theorem~\ref{thm:main} holds in this setting as well; the case $\mathbb{F} = \mathbb{C}$ of \cite[Definition 7]{allident} is recovered.
    \end{remark}
\begin{theorem}[Main Result]\label{thm:main}
We have
\[
L_\mathrm{inj}  \leq_p L_{\mathrm{gid},1}\leq_p L_\mathrm{inj\_det}\quad\text{and}\quad L_\mathrm{inj}  \leq_p L_\mathrm{gioid}
\leq_p 
L_\mathrm{inj\_det}. \]
\end{theorem}
\begin{proof}
Follows from Lemmas~\ref{thm:gid_hard} and ~\ref{lem:upper}.
\end{proof}

\subsection{Lower  bound for identifiability and IO-identifiability}\label{sec:lower} 
\begin{lemma}\label{thm:gid_hard}
If $\Char\mathbb{F} =0$ or $\mathbb{F}$ is finite, then
\[
L_\mathrm{inj}  \leq_p L_\mathrm{gid}\quad\text{and}\quad L_\mathrm{inj}  \leq_p L_\mathrm{gioid}. 
\]
\end{lemma}
\begin{proof}
Given a rational mapping 
\[
{\bm g}:\bm\alpha \mapsto (g_1(\bm\alpha),\ldots,g_n(\bm\alpha)), 
\]
where $g_1,\ldots,g_n \in  \mathbb{F}[\alpha_1,\ldots,\alpha_n]$, 
we can form the following linear ODE system
\begin{equation}\label{eq:Sigma}
  \Sigma_{\bm g} := \begin{cases}
        \dot{x}_{1} = x_2  \\ 
        \vdots\\
      \dot x_{n-1} = x_n \\
       \dot{x}_n = g_1(\bm\alpha)x_1+\ldots+g_n(\bm\alpha)x_n \\
       y = x_1  \\
  \end{cases}  ,
\end{equation}
in polynomial time. The proof for $L_\mathrm{inj}  \leq_p L_\mathrm{gid}$ is completed by showing that the mapping 
$$
F : \bm g \mapsto \Sigma_{\bm g}
$$
is a polynomial-time Karp reduction between $L_\mathrm{inj}$ and $ L_\mathrm{gid}$; in particular, we show that 
$$
\bm g \in L_\mathrm{inj}  \iff \Sigma_{\bm g} \in L_\mathrm{gid}. 
$$
By successive differentiation, it follows from~\eqref{eq:Sigma} that
\begin{equation}\label{eq:Y}
\begin{aligned}
y^{(n)} &= g_1(\bm\alpha)y +\ldots+ g_n(\bm\alpha)y^{(n-1)},\\
&\quad\vdots\\
y^{(2n-1)} &= g_1(\bm\alpha)y^{(n-1)} +\ldots+ g_n(\bm\alpha)y^{(2n-2)}.
\end{aligned}
\end{equation}
Denote $\bm\beta = \bm g(\bm\alpha)$
 and consider~\eqref{eq:Y} as a system of linear equations in $\beta_1,\ldots,\beta_n$. By \cite[Lemma~1]{OPT19}, whose proof generalizes to still be correct if we replace $\mathbb{C}$ by an arbitrary differential field of constants $\mathbb{F}$, the determinant of system~\eqref{eq:Y}\begin{equation}\label{eq:Wr}
 \Wr\left(y,\ldots,y^{(n-1)}\right)\notin\left[y^{(n)}-\left(\beta_1 y +\ldots+ \beta_ny^{(n-1)}\right)\right] = I_\Sigma\cap\mathbb{F}(\bm\beta)\{y\},
 \end{equation}
where the latter equality of differential ideals is true by a computation eliminating the $\bm x$-variables in~\eqref{eq:Sigma}. 
Therefore, modulo $I_\Sigma\cap\mathbb{F}\{y\}$ (which is the same as  ``for every generic solution $(\bm x^\ast,y^\ast)$ of~\eqref{eq:Sigma}''), we have 
\begin{equation}\label{eq:inRy}
g_1(\bm\alpha),\ldots,g_n(\bm\alpha) \in \mathbb{F}\langle y^\ast\rangle.
\end{equation}
Suppose that $\bm\alpha$ is not identifiable, and so let $i$ be such that $\alpha_i$ is not identifiable. By Definition~\ref{def:id}, this implies that for all $h \in \mathbb{F}\langle y^\ast\rangle\cap \mathbb{F}(\alpha_i)$, the map $h : \mathbb{F}\to\mathbb{F}$ is not injective on $D_i$. With this,~\eqref{eq:inRy} implies that the map $\bm{g}$ is not injective on $\bm D$.

Suppose now that the parameters $\bm \alpha$ are identifiable, and so let $(\bm x^\ast,y^\ast)$ be any generic solution of~\eqref{eq:Sigma} and, for all $i$,
\begin{equation}
\label{eq:alpha}
H_i := \cfrac{\bm h_{1,i}(y^\ast, {y^\ast}',\ldots)}{\bm h_{2,i}(y^\ast, {y^\ast}',\ldots)}\in\mathbb{F}(\alpha_i) \text{ is injective as a map }D_i \to \mathbb{F},
\end{equation}
where $\bm h_{1,i}, \bm h_{2,i} \in \mathbb{F}{\{y\}}$ and none of  $\bm h_{2,i}$ is in $I_\Sigma\cap\mathbb{F}(\bm\alpha)\{y\}$. 

Denote $\bm H = (H_1,\ldots,H_n)$ and consider the following two subfields: $\mathbb{F}(\bm\beta) \subset \mathbb{F}(\bm\alpha)$ and  $\mathbb{F}(\bm H) \subset \mathbb{F}(\bm\alpha)$. We will show that  $\mathbb{F}(\bm H)$ is contained in the perfect closure $\mathbb{K}$ of $\mathbb{F}(\bm\beta)$. Note that, if $\Char\mathbb{F} = 0$, then \[\mathbb{K} = \mathbb{F}(\bm\beta) = \mathbb{F}(\bm g(\bm\alpha)),\] and, if $\mathbb{F}$ is a finite field with $\Char\mathbb{F} = p$, then \[\mathbb{K} = \mathbb{F}\left(\beta_i,\beta_i^{1/p},\beta_i^{1/p^2},\ldots\mid 1\leqslant i\leqslant n\right) = \mathbb{F}\left(g_i(\bm\alpha),g_i(\bm\alpha)^{1/p},g_i(\bm\alpha)^{1/p^2},\ldots\mid 1\leqslant i\leqslant n\right).\] Suppose the contrary, $\bm H \notin \mathbb{K}$.

Let $\overline{\mathbb{F}(\bm\alpha)}$ denote the algebraic closure of $\mathbb{F}(\bm\alpha)$. 
By \cite[Theorem~9.29]{MilneFT}, there is 
$\varphi \in \mathrm{Aut}\left(\overline{\mathbb{F}(\bm\alpha)} / \mathbb{K} \right)$ 
 that satisfies $\varphi(\bm H)\ne \bm H$.
Let $j$, $1\leqslant j\leqslant n$, be such that $\varphi(H_j) \ne H_j$.
We perform an extension of scalars and complete the proof similarly to \cite[proof of Theorem~1]{ident-compare};~\eqref{eq:alpha} and the right-hand side of~\eqref{eq:Wr} imply that there exists $p_j \in \mathbb{F}(\bm\alpha)\{y\}$ such that 
\begin{equation}\label{eq:halphabeta}
\bm{h}_{2,j}(y)\cdot H_j- \bm{h}_{1,j}(y) = p_j\cdot \left(y^{(n)}-\beta_1 y -\ldots-\beta_ny^{(n-1)}\right).
\end{equation}
We extend $\varphi$ to $\overline{\mathbb{F}(\bm\alpha)}\{y\}$ by letting $\varphi(y) = y$. Applying $\varphi$ to~\eqref{eq:halphabeta} and subtracting~\eqref{eq:halphabeta} from the result,
since $\bm h_{1,j}, \bm h_{2,j}\in\mathbb{F}\{y\}$, we obtain
\begin{equation}\label{eq:subtract}
\bm h_{2,j}\cdot (\varphi(H_j) - H_j)
 = (\varphi(p_j)-p_j)\cdot\left(y^{(n)}-\beta_1 y -\ldots-\beta_ny^{(n-1)}\right).
 \end{equation}
Since $I_\Sigma$ is generated by linear forms, it is  prime under extensions of scalars, and so the ideal \[P := \left(\overline{\mathbb{F}(\bm\alpha)} \otimes I_\Sigma \right) \cap \overline{\mathbb{F}(\bm\alpha)}\{y\}\] is prime.
 Since $\bm h_{2,j} \notin P$ and $ 0\ne\varphi(H_j) - H_j \in \overline{\mathbb{F}(\bm\alpha)}$ and so is not in $P$, the primality of $P$ contradicts~\eqref{eq:subtract} as the right-hand side of~\eqref{eq:subtract} is in $P$. Thus, the above $\varphi$ cannot exist and therefore \[\mathbb{F}(\bm H) \subset \mathbb{K}.\] Since $\bm H$ is injective on $\bm D$ and is a composition of $\bm g(\bm\alpha)$ with another rational map in characteristic zero or with a rational map and the injective maps of taking $p$th roots if $F$ is finite of characteristic $p$, we finally conclude that $\bm g(\bm\alpha)$ is injective on $\bm D$. 
Thus, summarizing, $\bm{\alpha}$ is identifiable on $\bm D$ if and only if $\bm{g}$ is injective on $\bm D$, and so $L_\mathrm{inj}  \leq_p L_\mathrm{gid}$. From here, we also conclude that $\bm\alpha$ is identifiable on $\bm D$ if and only if $\bm\alpha$ is IO-identifiable on $\bm D$ because
\[
y^{(n)} = g_1(\bm\alpha)y +\ldots+ g_n(\bm\alpha)y^{(n-1)}
\]
is the IO-equation for $\Sigma_{\bm g}$ (cf. \cite[Proposition~1]{OPT19}). This  implies that $L_\mathrm{inj}  \leq_p L_\mathrm{gioid}$ and finishes the proof. 
 \end{proof}

\begin{remark}
    Notice that the reduction $
F : \bm g \mapsto \Sigma_{\bm g}
$ is polynomial whether we consider bit or algebraic complexity since the reduction itself is an almost trivial transformation.
\end{remark}

\subsubsection{Cubes and Knapsacks}

\begin{definition}\label{def:cube}
We call a set $C \subset R^n$ an $n$-{\em cube} if there exist $a_{1,0},...,a_{n,0},a_{1,1},...,a_{n,1}$ such that  
        $$
C := \{ (a_{1,s_1}, a_{2,s_2} ,...,a_{n,s_n}) \in R^n \mid s \in \{0,1\}^n \}
$$ 
where 
        $
a_{i,0} \neq a_{i,1}
        $
        for all $i \in \{0,1\}$.
        We call $C(0^n,1^n):=\{0,1\}^n$ the Boolean cube.
    
  Given some $S \subset \mathbb{Z}$,  we say that $D \subset R^n$ {\em contains a real $n$-cube over $S$} if:
    \begin{enumerate}
        \item It contains an $n$-cube $C(\bm a_{0},\bm a_{1})$
        \item The following weakening/generalization of being formally real holds (here $S = \{s_1,\ldots, s_{n+1}\} \subset \mathbb{Z}$ is a fixed knapsack instance as in~\eqref{eq:3}, so the integers $s_j$ appearing below are its entries): 
       \begin{multline} \label{eq:11}
\sum_{i \in [n]}(b_i-a_{i,0})^2(b_i-a_{i,1})^2 +\left(s_{n+1}-\sum_{j \in [n]}\frac{s_j(b_j-a_{j,0})}{(a_{j,1}-a_{j,0})} \right)^2 = 0 \\ 
\implies (b_1,...,b_n) \in \{ (a_{1,s_1}, a_{2,s_2} ,...,a_{n,s_n}) \mid s \in \{0,1\}^n \}. 
       \end{multline} 
        for all $(b_1,...,b_n) \in D$.
        Equivalently,~\eqref{eq:11} states that a finite sum of squares of elements of $R$ vanishes only if each summand vanishes, which is exactly the defining property of a formally real field. When $R$ is formally real,~\eqref{eq:11} is automatic. When $R$ is not formally real (in particular in non-zero characteristic, e.g. $R = \mathbb{F}_p$, where $1^2 + \ldots + 1^2 = 0$ for $p$ summands),~\eqref{eq:11} is a genuine restriction on $D$; this is precisely why Definition~\ref{def:knap} introduces the knapsack variety $V_{S,a}$ as the largest subset of the cube on which~\eqref{eq:11} is forced to hold, and why Lemma~\ref{lem:full_dim} distinguishes the ordered-ring case (where $V_{S,a} = R^n$) from the general case.
    \end{enumerate}
\end{definition}

\begin{remark}
    If $R$ is formally real, then the second condition is satisfied for any $S$. It is straightforward to prove that the module $R^n$ contains the Boolean cube $\{0,1\}^n$, when $R$ is a non-trivial ring; furthermore, an $n$-cube contains a real $n$-cube by definition. The motivation behind the definition is to allow for the generalization of Lemma~\ref{lem:inj} to: 
    \begin{enumerate}
        \item More general codomains other than $R^n$ and  $\{0,1\}^n$ 
        \item More general fields that are not necessarily formally real. 
    \end{enumerate}
    To that end, we define the largest variety that contains a a real $n$-cube in Definition~\ref{def:knap}. 
\end{remark}

\begin{definition}(Knapsack Variety)\label{def:knap}
Given a cube $C(\bm a_{0},\bm a_{1})$ and a set $S \subset \mathbb{Z}$ we define the {\em knapsack variety} over $S,\bm a$ as 
\begin{equation}\label{eq:knap}
 V_{S,\bm a} :=   \left( C(\bm a_{0},\bm a_{1}) \setminus  Z\left( p_{S,\bm a} ( \bm x)\right) \right) \cup T ,
\end{equation}
where 
\begin{equation}\label{eq:knap_poly}
    p_{S,\bm a}(\bm x) := \sum_{i \in [n]}(x_i-a_{i,0})^2(x_i-a_{i,1})^2 +\left(s_{n+1}-\sum_{j \in [n]}\frac{s_j(x_j-a_{j,0})}{(a_{j,1}-a_{j,0})} \right)^2 ,
\end{equation}
and 
\begin{equation}\label{eq:add_back_in}
    T:= \left\{ x\in C(\bm a_{0},\bm a_{1})  \mid s_{n+1}-\sum_{j \in [n]}\frac{s_j(x_j-a_{j,0})}{(a_{j,1}-a_{j,0})} = 0 \right\}.
\end{equation} 
Intuitively, the knapsack problem given by $S:= \{s_1,....,s_{n+1}\}$ is asking whether a subset of $S\setminus \{ s_{n+1}\}$ can add up to $s_{n+1}$. 
\end{definition}

\begin{lemma}\label{lem:full_dim}
   The knapsack variety given by Definition~\ref{def:knap} is an $n$-dimensional polynomial time definable sub-variety of $R^n$. In particular, it is a semi-algebraic set in the case that $R$ is an ordered ring and it is a Zariski open subset of $R^n$ for more general rings. In the case of ordered rings, we further have the equality  
   \begin{equation*}
        V_{S,\bm a} = R^n.
   \end{equation*}
\end{lemma}

\begin{lemma}\label{lem:inj}
Injectivity is NP-hard over real $n$-cubes.    
\end{lemma}
\begin{proof}
Given a set $S \subset \mathbb{Z} $ of size $n+1$ 
we can construct the cube $C(\bm a_0,\bm a_1)=C(0^n,1^n) $ and the polynomial $p_{S,\bm a}(x)$ from~\eqref{eq:knap_poly} in polynomial time.
Consider the polynomial mapping $g: V_{S,\bm a} \times R \rightarrow R^{n+1}$ defined coordinate-wise by
\begin{equation}
    g_i(x) = \begin{cases}
       x_i & \text{ if } i \in [n] \\ 
       x_{n+1}p_{S,\bm a}(x_1,...,x_n) & \mathrm{otherwise} \\ 
    \end{cases}.
\end{equation}
This construction is similar to the proof from \cite{Balreira2014}; in particular, 
it suffices to prove that 
\begin{equation}
    S \in L_\mathrm{knap} \iff g \not\in L_\mathrm{inj}.
\end{equation}
Equations~\eqref{eq:knap},~\eqref{eq:knap_poly}, and~\eqref{eq:add_back_in} 
give us that an element $x \in R^n$ which satsfies $ p_{S,\bm a} ( \bm x) = 0$ is in $V_{S,\bm a}  $ if and only if it satisfies \[s_{n+1}-\sum_{j \in [n]}\frac{s_j(x_j-a_{j,0})}{(a_{j,1}-a_{j,0})} = 0;\] indeed, by construction we took out all of the elements that satisfy $ p_{S,\bm a} ( \bm x) = 0$ (i.e., the intermediate set $R^n \setminus  Z\left( p_{S,\bm a} ( \bm x)\right)$) and then put the elements that satisfy the knapsack problem back in (\emph{i.e.,} the set $T$).
The proof can now be completed similarly to \cite{Balreira2014}:
Suppose that $S \in L_\mathrm{knap}$, then it is easy to see that there exists some $x \in C(0^n,1^n)$ such that \[g(x_1,...,x_n, 0 ) =g(x_1,...,x_n, 1 )  \] (since they both take the value $(x_1,...,x_n, 0 )$) and thus $g$ is not injective. 
Similarly, if $S \not\in L_\mathrm{knap}$, then there can be no assignment $x$ such that \[g(x_1,...,x_n, 1 ) = (x_1,...,x_n,0);\] therefore, $g$ is injective in this case. 
\end{proof}

\begin{remark}
    One may wonder why we used such a complicated definition of $p_{S,a}$ since we could used the simpler one from \cite{Balreira2014}. 
    The proof works just as well if you use any $C(a_0,a_1)$ and thus we have that any Domain that contains a combinatorial cube is complicated enough to have injectivity become an NP-hard problem; furthermore, the (lower bound on the) complexity seems to grow with exactly the dimension of the largest dimension of such a cube. 
    The proof of \cite{Balreira2014} empahsizes that this complexity is a consequence of being like the reals, but actually it can be made to work for more general fields (as we just did) and it seems to be more of a consequence of containing a large cube. 
\end{remark}

\begin{definition}
    \label{def:alg_lang}
    We define the following subsets of $L_\mathrm{inj}$: 
    \begin{equation}
        L_\mathrm{reg\_map} := \{(g,D) \in L_\mathrm{inj} \mid D \text{ a Zariski open subset of }R^n \text{ for some } n\},
    \end{equation}
    \begin{equation}
        L_\mathrm{semi\_alg} := \{(g,D) \in L_\mathrm{inj} \mid D \text{ a semi-algebraic subset of }R^n \text{ for some } n\},
    \end{equation}
    and
    \begin{equation}
        L_\mathrm{inj\_cube} := \{(g,D) \in L_\mathrm{inj} \mid D \text{ contains a real $n$-cube over $S,\bm a$ for some } S, \bm a \}.
    \end{equation}
\end{definition}

\begin{remark}
    In the case of algebraically closed fields we famously have 
    \begin{equation}
        L_\mathrm{reg\_map} = L_\text{inj}
    \end{equation}
    and in the case of real closed fields we famously have 
    \begin{equation}
        L_\mathrm{semi\_alg} = L_\text{inj}.
    \end{equation}
    Furthermore, by quantifier elimination, it is straightforward to prove that in both of these cases, we have that $L_\text{inj}$ could have been defined by simply giving $D$ as a set of polynomial equations (or inequalities) which is polynomial times definable in a straightforward way. 
\end{remark}

\begin{corollary}\label{corr:cube_hard}
    Semi-algebraic sets, (Zariski) open subsets of $F^n$ for $F$ a field, and, more, generally, sets definable in either a minimal or $o-$minimal (ring) structure are all NP hard to test injectivity on.  In particular, we have the following generalization of Lemma~\ref{lem:inj}: if a language $L$ satisfies 
    \begin{equation*}
        L_\mathrm{inj\_cube} \subset L \subset L_\mathrm{inj},
    \end{equation*}
    then it is NP hard. 
\end{corollary}

\subsection{Upper bound for complexity of identifiability}\label{sec:upper}
It will become useful to define an ``unexpanded'' language of injective determintal polynomial maps as 
\[
L_\mathrm{inj\_unexp}^{(n)} : = \left\{ (\bm g, D) \in L_\mathrm{inj}  \vert g_{i,j}  \in L_\mathrm{inj\_unexp}^{(n-1)}, \bm{g}^{\circ}: D \rightarrow \mathbb{F}^n \text{ is injective, } D \in L_\text{poly\_def} \right\} ,
\]
where we define $\bm{g}^{\circ}$ as the vector whose $i^\mathrm{th}$ coordinate is equal to $g_i = g_{i,0}( g_{i,1}, ..., g_{i,\ell_{i}})$ where $\ell_i$ is the number of variables in $g_{i,0}$ and $L_\mathrm{inj\_unexp}^{(0)} = L_\mathrm{inj} $.
See Example~\ref{ex:full} for why this is necessary. Lemma~\ref{lem:not_too_hard} proves that it suffices to use this more general definition for the upper bound.

\begin{lemma}\label{lem:upper} We have \[
L_{\mathrm{gid},1}  \leq_p L_\mathrm{inj\_det}\quad\text{and}\quad L_\mathrm{gioid}  \leq_p L_\mathrm{inj\_det}. 
\]
\end{lemma}
\begin{proof}
Given a linear system $\Sigma$, augment it, in polynomial time 
with the derivatives of all of the $\bm{x}$-equations up to order $n-1$ and of the $\bm{y}$-equations up to order $n$. The new system $\Sigma^{(n)}$ contains 
$n(n+m)+m$
equations. Again in polynomial time, 
we now eliminate the $\bm{x},\ldots,\bm{x}^{(n)}$-variables in system $\Sigma^{(n)}$ 
using unexpanded Gaussian Elimination to perform implicitization as described in \cite{Cox1997IdealsVA}. For the elimination, we order our main variables as follows:
\[\bm{x}^{(n)},\ldots,\bm{x},\bm{y}^{(n)},\ldots,\bm{y}.\]
The total number of variables to be eliminated is $n(n+1)$, which is stricly less than the number of equation $n(n+m)+m$ for all $m\geqslant 1$. In the language of differential algebra, we perform this elimination using a block orderly-elimination ranking on $\bm{x},\bm{y}$.

In the elimination, if the original system $\Sigma$ had coefficients $f_{i,j}({\bm \alpha}),g_{i,j}({\bm \alpha})$ for the $x_i$, then we treat the $f_{i,j},g_{i,j}$ as variables in the symbolic Gaussian Elimination so that we are left with equations of the form
\begin{equation*}
    E_{\ell,k} := h_{\ell,k}(f_{i,j},g_{i,j}) y^{(n_{\ell,k})}_{\ell}+\ldots.
\end{equation*}
It is straightforward to verify that this unexpanded implicitization will be polynomial time, since we leave the $h_{\ell,i,j}$ unexpanded and therefore each $h_{\ell,i,j}$ is a linear-size rational expression in the $f_{i,j},g_{i,j}$, which takes at most $O(n^3)$ steps to compute for each $y_{\ell}$ (see the discussion in Example~\ref{ex:full} for an explanation as to why this is necessary).

To obtain IO-equations $\bm{E} = E_1',\ldots,E_m'$, we select, for each $i$, $1\leqslant i\leqslant m$, the equation in $\{E_{i,1},...,E_{i,n_i}\}$ whose leading variable is a derivative of $y_i$ and the order of this derivative is the smallest among all equation with leading variable being a derivative of $y_i$. This can again be done in polynomial time. Notice that the leading coefficients in $\bm{E}$ can all be made equal to $1$ in polynomial time too. Because the unexpanded Gaussian elimination is equivalent to computing a characteristic presentation with respect to a block orderly ranking, the proof of \cite[Proposition~1]{OPT19} implies that these selected minimal-order equations are exactly the  IO-equations for the system.

For each $i$, $1\leqslant i \leqslant m$, let $\bm{g}_i$ denote the tuple of coefficients of $E_i'$.
Let \[\psi: \bm\alpha \mapsto (\bm{g}_1(\bm\alpha),\ldots,\bm{g}_m(\bm\alpha)).\] Denote $\bm{\beta}_i = \bm{g}_i(\bm\alpha)$, $1\leqslant i \leqslant m$. By the proof of \cite[Proposition~1]{OPT19}, 
\begin{equation}\label{eq:Rbetaisk}
\mathbb{F}(\bm{\beta}_1,\ldots,\bm{\beta}_m)=k,\end{equation}
where the field $k$ for $\Sigma$ is as in Definition~\ref{item:eight}.  Suppose $\bm\alpha$ is not IO-identifiable over $\bm D$, and let $i$ be such that $\alpha_i$ is not IO-identifiable. By definition, for all $h \in k\cap\mathbb{F}(\alpha_i)$, the map $h:D_i\to\mathbb{F}$ is not injective. In particular, $\psi$ is not injective on $\bm D$. On the other hand, if $\psi$ is not injective on $\bm D$, then, by~\eqref{eq:Rbetaisk}, there exists $i$ such that, for all $h\in k\cap\mathbb{F}(\alpha_i)$, the map $h:D_i\to\mathbb{F}$ is not injective. Therefore, $\alpha_i$ and so $\bm\alpha$ are not IO-identifiable on $D_i$ and $\bm D$, respectively.

Thus, the parameters $\bm{\alpha}$ are IO-identifiable on $\bm D$ if and only if $\psi$ is injective on $\bm D$,  showing 
\begin{equation}\label{eq:GIOID}
L_\mathrm{gioid}  \leq_p L_\mathrm{inj\_det}.
\end{equation}

In regards to $L_{\mathrm{gid}}$, let $m=1$, that is, we only have one output. Let  $\Char\mathbb{F} = p$ be arbitrary. Then, by \cite[Theorem~1]{OPT19} with $\mathbb{C}$ replaced by $\mathbb{F}$ everywhere there (following its proof as in the proof of Lemma~\ref{thm:gid_hard}), in the notation of~\eqref{eq:sys} and Definition~\ref{def:id}, the field of IO-identifiable functions is equal to
\begin{equation}\label{eq:IOandnotIO1}
\mathbb{F}\langle y^\ast\rangle^{p^{-\infty}}\cap \mathbb{F}(\bm\alpha).\end{equation}
Comparing Definition~\ref{def:id} and Definition~\ref{item:eight}, we conclude using~\eqref{eq:IOandnotIO1} that $\bm\alpha$ is identifiable on $\bm D$ if and only if $\bm\alpha$ is IO-identifiable on $\bm D$.
Thus,  we obtain $L_{\mathrm{gid},1}$ is equivalent to $L_{\mathrm{gioid},1}$. Therefore, by~\eqref{eq:GIOID}, we obtain $L_{\mathrm{gid},1}  \leq_p L_\mathrm{inj\_det}$. 
\end{proof}

\begin{example}\label{ex:full} We will see that a straightforward way of reducing IO-identifiability to the injectivity of the coefficient map of IO-equations (by doing Gaussian elimination of the state variables to find the IO-equations) is not necessarily polynomial. 
For this, consider the following (dense, generic) linear ODE model:
\begin{equation}\label{eq:full}
  \begin{cases}
        \dot{x}_{1} = a_{11}x_1+\ldots+a_{1n}x_n  \\ 
        \  \quad\vdots\\
       \dot{x}_n = a_{n1}x_1+\ldots+a_{nn}x_n \\
       y = x_1  \\
  \end{cases}  
\end{equation}
A calculation shows (see the proof of \cite[Theorem~2]{MSE}) that the coefficient of $y$ in the IO-equation for~\eqref{eq:full} is the determinant of the matrix \[
A = \begin{pmatrix}
a_{11}&\ldots& a_{1n}\\
\vdots&\ddots&\vdots\\
a_{n1}&\ldots& a_{nn}
\end{pmatrix},
\]
which if expanded has $n!$ terms, and so the corresponding injectivity problem has input of $n!$ size. Therefore, since, in the input of the injectivity problem, the polynomials are expanded,  this particular way of reducing the IO-identifiability problem to the injectivity problem is not polynomial in the input size (so, we cannot claim that  $L_\mathrm{gioid}  \leq_p L_\mathrm{inj}$). 
\end{example}
\begin{remark}
In regards to $L_{\mathrm{gid}}$, we would like to add the following. Let $m=1$ and $\Char F = 0$, that is, we only have one output. Then, by \cite[Theorem~1]{OPT19} with $\mathbb{C}$ replaced by $\mathbb{F}$ everywhere there, in the notation of~\eqref{eq:sys} and Definition~\ref{def:id}, the field of IO-identifiable functions is equal to
\begin{equation}\label{eq:IOandnotIO}
\mathbb{F}\langle y^\ast\rangle\cap \mathbb{F}(\bm\alpha).\end{equation}
Comparing Definition~\ref{def:id} and Definition~\ref{item:eight}, we conclude using~\eqref{eq:IOandnotIO} that $\bm\alpha$ is identifiable on $\bm D$ if and only if $\bm\alpha$ is IO-identifiable on $\bm D$.
Thus,  we obtain $L_{\mathrm{gid},1}$ is equivalent to $L_{\mathrm{gioid},1}$. If $\mathbb{F}$ is a finite field, then, following the proof of Lemma~\ref{thm:gid_hard}, we can show that the field of globally identifiable functions~\eqref{eq:IOandnotIO} is contained in the perfect closure of the field of IO-identifiable functions. 
\end{remark}

\begin{lemma}\label{lem:not_too_hard}
    For any ring $R$, we have 
    \begin{equation*}
        L_\mathrm{inj\_unexp} \in \mathbf{co\text{-}NP}_R.
    \end{equation*}
\end{lemma}

\begin{proof}
We prove the existence of a certificate and a polynomial time verification machine for 
$(g,D)$, if it is not injective. 
If $g$ is not injective on $D$, then there are points $\bm p, \bm q \in D$ such that $\bm p \neq \bm q$ and $g(\bm p) = g(\bm q)$. We take as certificate the points $\bm p, \bm q \in D$ and as verifier a machine that computes the boolean function $g(\bm p) = = g(\bm q)$ 
We now analyze the complexity: 
\begin{itemize}
    \item (\textbf{Polynomial Certificate Size}) $\mathrm{length}(\bm p),\mathrm{length}(\bm q)$ are both polynomial size in terms of the input size $|(g,D)|$; in particular, if $g_{i,j} \in \mathbb{F}[\alpha_1,...,\alpha_s] $, then $\bm p  , \bm q \in \mathbb{F}^s$. Therefore, since elements of the alphabet $\mathbb{F}$ are considered to have unit cost, we have that  $\mathrm{length}(\bm p),\mathrm{length}(\bm q) = O(s) = O(|(g,D)|)$ as needed. 
    
    \item (\textbf{Polynomial Verifier}) Verifying whether $g(\bm p) = g(\bm q)$ can be done in polynomial time in terms of the input size $|(g,D)|$. 
    We perform an induction on the smallest $n$ for which $(g,D) \in L_\mathrm{inj\_unexp}^{(n)} $.
    The base case where $(g,D) \in L_\mathrm{inj\_unexp}^{(1)}$ 
    is simply evaluating a polynomial, which is straightforward to see can be done in ploynomial time (the number of aritmetic operations is directly proportionate to $\mathrm{length}(\bm p)$). 
Notice that if $(g,D) \in L_\mathrm{inj\_unexp}^{(n)} $, then evaluating $(h\circ f)({\bm p}) := g({\bm p})$ for some ${\bm p} \in D, (f, D) \in L_\mathrm{inj\_unexp}^{(n-1)}$, and $(h,\mathrm{range}(f)) \in L_\mathrm{inj\_unexp}^{(1)}$ is the same thing as evaluating $f({\bm p})$ and then $h(f({\bm p}))$. 
The induction hypotheseis gives us that computing $f({\bm p})$ can be done in polynomial time and $h(f( {\bm p}))$ is simply the base case above. 
Finally after computing both $g({\bm p})$ and $g({\bm q})$ checking whether they are equal has unit cost. 
    
\end{itemize}
\end{proof}

\begin{corollary}\label{cor:final} Let $R$ be a ring in the sense of the BSS model (Definition~\ref{def:5}). Then
\begin{equation}\label{eq:GIOID_conp}
L_\mathrm{gioid}   \in \mathbf{co\text{-}NP}_R \quad\text{and}\quad L_\mathrm{gid,1} \in  \mathbf{co\text{-}NP}_R .
\end{equation}

\end{corollary}
\begin{proof}
By Lemma~\ref{lem:upper}, both $L_{\mathrm{gioid}}$ and $L_\mathrm{gid,1}$ are polynomial-time Karp reducible to $L_\mathrm{inj\_det}$, which is the depth-one case of $L_\mathrm{inj\_unexp}$. By Lemma~\ref{lem:not_too_hard}, $L_\mathrm{inj\_unexp}\in \mathbf{co\text{-}NP}_R$. Since $\mathbf{co\text{-}NP}_R$ is closed under polynomial-time Karp reductions, the result follows. 
\end{proof}

\bibliographystyle{abbrvnat}
\bibliography{bib.bib}

\end{document}